\documentclass[conference, 10pt, twocolumn]{ieeeconf}

\IEEEoverridecommandlockouts
\usepackage[usenames]{color}
\usepackage{enumerate}
\usepackage{url}
\usepackage{subfigure}
\usepackage{amsfonts,mathrsfs}
\usepackage{amssymb,amsmath}
\usepackage{verbatim}
\usepackage{acronym}
\usepackage{mathtools}
\usepackage{graphicx}
\usepackage{kbordermatrix}

\def\fskip#1{}

\newtheorem{theorem}{Theorem}[section]

\newtheorem{definition}[theorem]{Definition}

\newtheorem{lemma}[theorem]{Lemma}

\newtheorem{proposition}[theorem]{Proposition}
\newtheorem{remark}[theorem]{Remark}

\renewcommand{\theenumi}{\arabic{enumi}}

\def\1{{\bf 1}}

\def\R{\mathbb{R}}

\newcommand{\remove}[1]{}

\renewcommand{\kbldelim}{(}
\renewcommand{\kbrdelim}{)}

\DeclarePairedDelimiter\floor{\lfloor}{\rfloor}

\begin{document}
\title{Convergence Time of Quantized Metropolis Consensus Over Time-Varying Networks}

\author{\authorblockN{Tamer Ba\c{s}ar, \ \ Seyed Rasoul Etesami, \ \ Alex Olshevsky}
  \authorblockA{}}
\maketitle
\begin{abstract}
We consider the quantized consensus problem on undirected time-varying connected graphs with $n$ nodes, and devise a protocol with fast convergence time to the set of consensus points. Specifically, we show that when the edges of each network in a sequence of connected time-varying networks are activated based on Poisson processes with Metropolis rates, the expected convergence time to the set of consensus points is at most $O(n^2\log^2 n)$, where each node performs a constant number of updates per unit time. 
\end{abstract}

\section{Introduction} 
There has been much recent interest in the design of control protocols for distributed systems, motivated in part  by the need to develop protocols for networks of autonomous agents characterized by the lack of centralized information and time-varying connectivity. A canonical problem within the field is the so-called average consensus problem, wherein a group of agents must agree on the average of their initial values while interacting with neighbors in a (possibly time-varying) network. Protocols for consensus problems must be distributed, relying only on local information at each node, and robust to unexpected changes in topology.

It is well understood by now that protocols for the consensus play an important role in a number of more sophisticated multi-agent tasks. We mention distributed optimization, coverage control, formation control, cooperative statistical inference, power control, load balancing, epidemic routing, as examples of control and coordination problems with proposed solutions relying crucially on consensus.

Our work is motivated by the observation that often working with real-valued variables in multi-agent control is neither necessary nor efficient. Indeed, limited memory and storage at each agent often forces the variables kept by each agent to lie in a discrete set. We therefore consider the quantized consensus problem, a variation on the consensus problem where the values of each agent are constrained to be integers lying within a certain range. Previous literature on this problem includes
\cite{Kashyab,etesami2014metropolis,zhu2011convergence,carli2010gossip,benezit2009interval, ishiiconv}. As some of the real-world applications of the problem under consideration, one can consider information fusion in
sensor networks \cite{xiao2007distributed}, e. g., when every sensor in a sensor
network has a different measurement of the temperature and
each sensor's goal is to compute the average temperature, and load balancing in processor networks which has various
applications in computer science.

The original paper \cite{Kashyab} contained upper bounds for a natural quantized consensus protocol on a variety of common graphs. A few years later, the work of \cite{zhu2011convergence} proposed a quantized consensus protocol with upper bound of
$O(n^4)$ on the expected convergence time for any fixed graph. For dynamic graphs, \cite{zhu2011convergence} obtained a convergence time scaling of $O(n^8)$. The paper \cite{ishiiconv} obtained an upper bound of $O(n^2)$, but only on complete graphs. Recently, it was shown in \cite{etesami2014arxiv-QC-static-dynamic} that a certain ``unbiased'' quantized consensus protocol has maximum expected convergence time of $O(nmD\log n)$ on static networks where $m$ and $D$ denote the number of edges and the diameter of the network, and has maximum expected convergence time $O(nm_{\max}D_{\max}\log^2n)$ on connected time-varying networks, where $m_{\max}$ and $D_{\max}$ denote, respectively, the maximum number of edges and the maximum diameter in the sequence of time-varying networks. Unbiasedness means that each step of the protocol was based on choosing a random edge uniformly among all possible edges in the network.

A faster upper bound on convergence time and only for static networks was given in the recent paper \cite{shang2013upper}, where a protocol was provided whose expected convergence time in general static networks is $O(n^2 \log n)$ \footnote{The bound given in \cite{shang2013upper} was $O(n^3 \log n)$, but this was a count of the total number of updates; in terms of time, this leads to a quadratic expected convergence time.}. As of the time of writing this paper, the upper bounds of $O(n^2\log n)$ and $O(nm_{\max}D_{\max}\log^2n)$, respectively, are the fastest protocols known to us for randomized quantized consensus over static and dynamic networks \cite{etesami2014arxiv-QC-static-dynamic}, \cite{shang2013upper}, respectively. It is worth noting that the convergence speed of a protocol is measured by the maximum over all initial inputs of the expected time that the given protocol will run until reaching consensus, where each node performs a constant number of updates per unit time.

In this paper, we analyze a protocol wherein nodes cooperate to perform updates on edges connecting them at so-called Metropolis rates. We are able to show that when each node performs a constant number of updates per unit time, convergence time on connected time-varying networks is at most $O(n^2 \log^2 n)$, giving us the fastest quantized consensus convergence guarantee over time-varying networks.

The paper is organized as follows.  In Section~\ref{sec:Problem-Formulation}, we  formulate the quantized consensus problem and discuss some of its properties. In Section~\ref{sec:mainresults}, we show that quantized consensus with Metropolis rates has essentially quadratic convergence time on fixed networks; this result is not better than the 
previous convergence time of \cite{shang2013upper}, but we include it here because our later results rely on it. In Section \ref{sec:time-varying}, we prove our main result, namely that Metropolis quantized consensus converges in essentially quadratic time on time-varying networks. We conclude the paper by identifying some future directions of research in Section~\ref{sec:conclusion}.

\textbf{Notation}:
We let $[n]=\{1,2,\ldots,n\}$. For an undirected graph $\mathcal{G}=(V,\mathcal{E})$ we let $N(x)$ be the set of neighbors of $x$, and furthermore, let $d(x)$ denote the degree of $x$ i.e., $d(x) = |N(x)|$. For a given vector $v$, we denote its $i$th entry by $v_i$ and its transpose by $v'$. We say that a matrix $A$ with nonnegative entries is stochastic if each row of $A$ sums to 1. We say $A$ is doubly-stochastic if both $A$ and $A'$ are stochastic. 

\section{Problem Formulation}\label{sec:Problem-Formulation} 
In this section, we assume that we are given a fixed, undirected, connected graph $\mathcal{G}=(V,\mathcal{E})$ with $|V|=n$ and without self-loops. The Metropolis Markov chain on this graph is then defined as follows.

Let $\mathcal{M}$ be a square matrix whose $ij$'th entry is defined as
\begin{align}\label{eq:Metropolis-chain}
\mathcal{M}_{ij} = \begin{cases} \frac{1}{\max\{d(i),d(j)\}}, & \mbox{if} \ i\neq j,~ j \in N(i)  \\
1-\sum_{j \in N(i)}\mathcal{M}_{ij}, & \mbox{if}  \ i=j \\
0 & \mbox{otherwise},
\end{cases}
\end{align} 
Note that $\mathcal{M}$ is  symmetric, nonnegative, and doubly stochastic. We refer to the Markov chain which transitions according to $\mathcal{M}$ as the \textit{Metropolis chain}. Moreover, given nodes $x,y \in V$, the \textit{hitting time} $H^{\rm m}(x,y)$ is the expected time until the Metropolis chain with initial position at node $x$ reaches node $y$ (quantities associated with the Metropolis chain will generally be denoted with an ``m'' superscript). For future use and convenience, we adopt the notation $\lambda_{ij} = 1/\max(d(i),d(j))$ whenever $\{i,j\} \in \mathcal{E}$ and $\lambda_{ij}=0$ otherwise.

\subsection{Quantized Metropolis dynamics}
We next introduce a continuous time quantized process based on Metropolis weights, whose behavior will turn out to be related to the Metropolis chain. For each link $\{i,j\}\in \mathcal{E}$ of the graph $\mathcal{G}$, we consider a Poisson process of rate $\lambda_{ij}$. Each node $i$ begins with an integer value $x_i(0)$ which lies in the range $[l,L]$. Each time the process corresponding to an edge $\{i,j\} \in \mathcal{E}$ registers an arrival, the two nodes $i,j$ perform the quantized consensus update from \cite{Kashyab}:
\begin{align}\label{eq:Quantize-rule}
x_i(t^{+}) = \begin{cases} x_i(t)-1, & \mbox{if} \ x_i(t)>x_j(t) \\
x_i(t)+1, & \mbox{if}  \ x_i(t)<x_j(t) \\
x_i(t), & \mbox{if}  \ x_i(t)=x_j(t),
\end{cases}
\end{align}
and likewise for node $j$. In other words, each $x_i(t)$ is an integer-valued jump process whose jumps occur whenever an arrival occurs at any
edge incident on $i$.

Note that the above update rule allows for the possibility that $x_i(t^+)=x_i(t)$. In this case, we will say that the update at time $t$ was {\em trivial}. Furthermore, observe that if $|x_i(t) - x_j(t)| = 1$, then the update of Eq. (\ref{eq:Quantize-rule}) will cause nodes $i$ and $j$ to swap values, i.e., $x_i(t^+)=x_j(t), x_j(t^+)  = x_i(t)$. In this case, we will also say that the update was trivial. If neither of these two cases has occurred during an update, we will say that the update was {\em non-trivial}. Simply speaking, a non-trivial update refers to the case where the incident nodes of an activated edge have integer values which differ by at least 2. 

It has been shown earlier in \cite{Kashyab} that such dynamics based on update rule \eqref{eq:Quantize-rule} will converge with probability 1 to consensus set defined by
\begin{align}\nonumber
\mathcal{C}=\Big\{x|x_i\!\in\! \{\floor{\bar{x}(0)},\floor{\bar{x}(0)}\!\!+\!\!1\}, \frac{\sum_{i=1}^{n}x_i}{n}=\bar{x}(0), i=\in [n]\Big\}.
\end{align}
In the remainder of this paper our goal is to provide an upper bound on the expected convergence time of the quantized Metropolis dynamics to the consensus set. 

\section{Convergence Time over Static Networks}\label{sec:mainresults}

We now begin the analysis of the convergence time of the quantized Metropolis algorithm over a fixed graph $\mathcal{G}$. The bounds derived in this section are of the same order of magnitude as the previous work \cite{shang2013upper}, and the proof follows similar lines as well; however we are including it here as our analysis over time-varying networks in the next section relies on these results. 

A key step in our analysis is to bound the expected time until the first nontrivial update takes place. For this purpose, we let $\overline{T}$ be the maximum expected time such that a nontrivial update takes place over all possible configuration of integers. More precisely, for integers $j_1, \ldots, j_n$ let $T(j_1, \ldots, j_n)$ be the expected time until a nontrivial update takes place when node $i$ begins with integer value $j_i$. Then  \[ \overline{T} = \max_{\substack{(j_1, \ldots, j_n) \in \mathbb{Z}^n, l\leq j_i\leq L \\ (j_1, \ldots, j_n) \notin \mathcal{C}}}~~ T(j_1, \ldots, j_n). \]

\begin{definition}  
Consider  two random walkers moving based on whether the activated edge in the quantized Metropolis dynamics is incident to them. That is, if one of the walkers is at node $x$, and if the next edge to register an arrival is incident to $x$, i.e., if it is $\{x,y\}$ for some $y\in N(x)$, then the random walker moves from $x$  to $y$.  We refer to such a process as the \textit{original} process. We denote by $M^{\rm o}(x,y)$ the expected ``meeting time'' of this process, defined to be the expected time until an edge incident to both walkers registers an arrival provided the two walkers started at nodes $x,y$ with $x \neq y$ (in general, we will denote quantities associated with this
process with an ``o'' superscript). By convention, we set $M^{\rm o}(x,x)=0$ for all $x \in V$.
\end{definition}

\begin{proposition}
$\overline{T} = \max_{x, y \in V} ~M^{\rm o}(x,y)$.
\end{proposition}
\begin{proof} 
The proof is immediate and is not included due to space limitation. 
\end{proof}

Based on the above proposition, our next step is to bound $\max_{x,y \in V} M^{\rm o}(x,y)$. We will actually find it easier to instead bound a meeting time associated with a slightly different process, which we call the \textit{virtual} process, defined next. 

\begin{definition}
The virtual process is identical to the original process until the two walkers $x,y$ become each other's neighbors in $\mathcal{G}$. At that time, the edge connecting them registers arrivals according to a Poisson process of rate $2 \lambda_{xy} $. We denote the expected meeting time function of the virtual process by $M^{\rm v}(x,y)$ (in general, we will denote quantities associated with this process with a ``v'' superscript). By convention we set $M^{\rm v}(x,x) = 0$ for all $x \in V$.
\end{definition}

\begin{definition}
A function $h:\Omega\rightarrow \mathbb{R}$ is called harmonic at a vertex $x\in \Omega$ for a Markov chain with transition probability matrix $P$ if $h(x)=\sum_{y\in \Omega}P(x,y)h(y)$.
\end{definition}

\begin{remark}\label{rem:harmonic-equality}
Given a nonempty subset $B\subset \Omega$ and a Markov chain with an irreducible transition matrix $P$, every harmonic function $h(\cdot):\Omega\rightarrow \mathbb{R}$ over $\Omega\setminus B$ which satisfies $h(x)\!\ge\! 0$, \!$\forall x\!\in\! B$, must be nonnegative over the entire $\Omega$ (see Proposition 9.1 of \cite{Markov-Book}).
\end{remark}

\begin{definition}
A vertex $\theta$ is called a hidden vertex of the Metropolis chain $\mathcal{M}$ if $H^{\rm m}(\theta,x)\leq H^{\rm m}(x,\theta), \forall x\in V$.
\end{definition}

\begin{remark}\label{rem:transitivity} 
It is known that the hitting times of a reversible Markov chain satisfy transitivity property, i.e., 
\begin{align}\nonumber 
\!\!\resizebox{0.01\hsize}{!}{$~$}\resizebox{1\hsize}{!}{$H^{\rm m}(x,y)\!+\!H^{\rm m}(y,z)\!+\!H^{\rm m}(z,x)\!=\!H^{\rm m}(x,z)\!+\!H^{\rm m}(z,y)\!+\!H^{\rm m}(y,x)$}, 
\end{align}see \cite{aldous2002reversible} for a proof. As a consequence every reversible Markov chain has at least one hidden vertex \cite{Fundamental}.
\end{remark}

In the following lemma we bound the expected meeting time of the virtual process by showing that the expected meeting time function of the virtual process $M^{\rm v}(x,y)$ satisfies essentially the same recursion as the function $\Phi^{\rm m}(\cdot,\cdot):V\times V \rightarrow \mathbb{R}$ defined by
\begin{align}\label{eq:phi-function}
\Phi^{\rm m}(x,y)=H^{\rm m}(x,y)+H^{\rm m}(y,\theta)-H^{\rm m}(\theta,y),
\end{align}
where $\theta$ is a fixed hidden vertex  of the Metropolis chain.

\begin{lemma}\label{thm:meeting-hitting-harmonic}
For all $x,y \in V$, we have $M^{\rm v}(x,y)\leq 6n^2$.
\end{lemma}
\begin{proof}
Fix nodes $x$ and $y$ such that $x \neq y$, and define $$\Lambda_x = \!\!\!\sum_{x_i \in N(x)} \!\!\!\lambda_{x x_i}, ~~ \Lambda_y = \!\!\!\sum_{y_i \in N(y)}\!\!\!\lambda_{y y_i}, ~~ \Lambda_{xy} = \Lambda_x + \Lambda_{y}.$$ An immediate consequence of Remark \ref{rem:transitivity} is the symmetry of $\Phi^m(\cdot, \cdot)$, i.e., $\Phi^{\rm m}(u,v)=\Phi^{\rm m}(v,u), \forall u,v \in V.$ Therefore we can argue that
\begin{small}
\begin{align}\nonumber
&\Phi^{\rm m}(x,y)=\frac{\Lambda_{x}}{\Lambda_{xy}}\Phi^{\rm m}(x,y)+\frac{\Lambda_{y}}{\Lambda_{xy}}\Phi^{\rm m}(y,x)\cr
&\!=\!\frac{\Lambda_{x}}{\Lambda_{xy}}\!\left[\frac{1}{\Lambda_x}\!+\!\!\!\!\!\sum_{x_i\in N(x)}\!\!\!\!\frac{\lambda_{xx_i}}{\Lambda_x}H^{\rm m}(x_i,y)\right]\!\!+\!\frac{\Lambda_x}{\Lambda_{xy}}\big[H^{\rm m}(y,\theta)\!-\!H^{\rm m}(\theta,y)\big]\!\cr
&\!+\!\frac{\Lambda_{y}}{\Lambda_{xy}}\left[\frac{1}{\Lambda_y}\!+\!\!\!\!\!\sum_{y_i\in N(y)}\!\!\!\!\frac{\lambda_{yy_i}}{\Lambda_y}H^{\rm m}(y_i,x)\right]\!\!+\!\frac{\Lambda_y}{\Lambda_{xy}}\big[H^{\rm m}(x,\theta)\!-\!H^{\rm m}(\theta,x)\big],
\end{align}\end{small}where in the second equality we expanded $H^{\rm m}(x,y)$ and $H^{\rm m}(y,x)$.
Simplifying the last relation, we obtain
\begin{footnotesize}
\begin{align}\label{Lemma:Phi-two}
&\Phi^{\rm m}(x,y)=\frac{2}{\Lambda_{xy}}+\!\!\sum_{x_i\in N(x)}\frac{\lambda_{xx_i}}{\Lambda_{xy}}H^{\rm m}(x_i,y)+\!\!\sum_{y_i\in N(y)}\frac{\lambda_{yy_i}}{\Lambda_{xy}}H^{\rm m}(y_i,x)\cr
&\qquad+\!\frac{\Lambda_x}{\Lambda_{xy}}\big(H^{\rm m}(y,\theta)\!-\!H^{\rm m}(\theta,y)\big)\!+\!\frac{\Lambda_y}{\Lambda_{xy}}\big(H^{\rm m}(x,\theta)\!-\!H^{\rm m}(\theta,x)\big),
\end{align}
\end{footnotesize}
Using the definition of $\Phi^{\rm m}(\cdot, \cdot)$, for each $x_i \in N(x)$ we have, $\Phi^{\rm m}(x_i,y)=H^{\rm m}(x_i,y)+H^{\rm m}(y,\theta)-H^{\rm m}(\theta,y).$ Multiplying this relation by $\frac{\lambda_{xx_i}}{\Lambda_{xy}}$ and summing over $x_i\in N(x)$, we obtain
\begin{align}\label{Lemma:Phi-three}
\!\!\!\!\!\!\!\!\!\sum_{x_i\in N(x)}\frac{\lambda_{xx_i}}{\Lambda_{xy}}\Phi^{\rm m}(x_i,y)&=\sum_{x_i\in N(x)}\frac{\lambda_{xx_i}}{\Lambda_{xy}}H^{\rm m}(x_i,y)\cr&+\frac{\Lambda_x}{\Lambda_{xy}}\big[H^{\rm m}(y,\theta)-H^{\rm m}(\theta,y)\big].
\end{align}
By symmetry of $\Phi^{\rm m}$ and using the same argument we get
\begin{align}\label{Lemma:Phi-four}
\!\!\!\!\!\!\!\sum_{y_i\in N(y)}\frac{\lambda_{yy_i}}{\Lambda_{xy}}\Phi^{\rm m}(x,y_i)&=\sum_{y_i\in N(y)}\frac{\lambda_{yy_i}}{\Lambda_{xy}}H^{\rm m}(y_i,x)\cr&+\frac{\Lambda_y}{\Lambda_{xy}}\big[H^{\rm m}(x,\theta)-H^{\rm m}(\theta,x)\big].
\end{align}
Substituting \eqref{Lemma:Phi-four} and \eqref{Lemma:Phi-three} in \eqref{Lemma:Phi-two} yields  \begin{small}
\begin{align}\label{Lemma:phi-harmonic}
\Phi^{\rm m}(x,y)=\frac{2}{\Lambda_{xy}}\!+\!\sum_{x_i\in N(x)}\!\!\frac{\lambda_{xx_i}}{\Lambda_{xy}}\Phi^{\rm m}(x_i,y)\!+\!\sum_{y_i\in N(y)}\!\!\frac{\lambda_{yy_i}}{\Lambda_{xy}}\Phi^{\rm m}(x,y_i).
\end{align} \end{small}On the other hand, we note that the meeting time of the virtual process can be expanded as  
\begin{small}
\begin{equation} \label{virtualmeet}
M^{\rm v}(x,y)=\frac{1}{\Lambda_{xy}}+\!\!\sum_{x_i\in N(x)}\!\!\frac{\lambda_{xx_i}}{\Lambda_{xy}}M^{\rm v}(x_i,y)+\!\!\sum_{y_i\in N(y)}\!\!\frac{\lambda_{yy_i}}{\Lambda_{xy}}M^{\rm v}(x,y_i).
\end{equation} 
\end{small}We therefore see that $\frac{1}{2}\Phi^{\rm m}(x,y)$ and $M^{\rm v}(x,y)$ follow the same recursion formula when $x \neq y$. Thus defining $f(x,y)=\frac{1}{2}\Phi^{\rm m}(x,y)-M^{\rm v}(x,y)$ we have $\forall x\neq y \in V$
\[ f(x,y) = \sum_{x_i \in N(x)} \frac{\lambda_{x x_i}}{\Lambda_{xy}} f(x_i,y) + \sum_{y_i \in N(y)} \frac{\lambda_{y y_i}}{\Lambda_{xy}} f(x, y_i).\] 
Defining the stochastic irreducible matrix $\mathcal{Q}$ as
\begin{align}\nonumber
\mathcal{Q}\big((x,y),(r,s)\big) = \begin{cases} \frac{\lambda_{xr}}{\Lambda_{xy}}, & \mbox{if} \ s=y,\ r\in N(x)  \\
\frac{\lambda_{yw}}{\Lambda_{xy}}, & \mbox{if} \ r=x, \ w\in N(y) \\
0, & \mbox{ Else},
\end{cases}
\end{align} 
one can see that $f(x,y)$ is harmonic for the matrix $\mathcal{Q}$ over $V\times V\setminus \{ (x,x) ~|~ x \in V\}$ and, since $\theta$ is a hidden vertex, we have that for all $x \in V$,  $f(x,x)=\frac{1}{2}(H^{\rm m}(x,\theta)-H^{\rm m}(\theta,x))\ge 0$. Therefore, using Remark \ref{rem:harmonic-equality}, we immediately get $f(x,y) \ge 0$ for all $x,y \in V$, and therefore $M^{\rm v}(x,y)\leq \frac{1}{2}\Phi^{\rm m}(x,y)$. Finally, it was shown in \cite{nonaka2010hitting} that $H^{m}(x,y) \leq 6 n^2$ for all $x,y \in V$, which now immediately implies the current lemma. 
\end{proof}

Our next lemma shows that $\max_{x,y} M^{\rm o}(x,y)$ is essentially of the same order as $\max_{x,y} M^{\rm v}(x,y)$.

\begin{lemma} \label{orig-lemma}
$\max_{x,y} M^{\rm o}(x,y) \leq 12 n^2$.
\end{lemma} 
\begin{proof} 
Fix $x\neq y$ and initialize both the original process and the virtual process with initial positions of the walkers being $x$ and $y$.
Furthermore, let us couple these processes to move identically until the first time when the walkers are each other's neighbors. At this time assuming that the walkers are at nodes $u$ and $w$, let us split the edge connecting them of
rate $2 \lambda_{uw}$ in the virtual process into two distinct edges of rate $\lambda_{uw}$ (Figure \ref{fig:virtual-original}). We will make the first of these (the bottom black edge in the right side of Figure \ref{fig:virtual-original}), as well as the edges going from $u$ and $w$ to their neighbors, register arrivals in the virtual process exactly when the corresponding edges in the original process register arrivals.
\begin{figure}[htb]
\vspace{-1.7cm}
\begin{center}
\includegraphics[totalheight=.15\textheight,
width=.25\textwidth,viewport=400 0 1300 900]{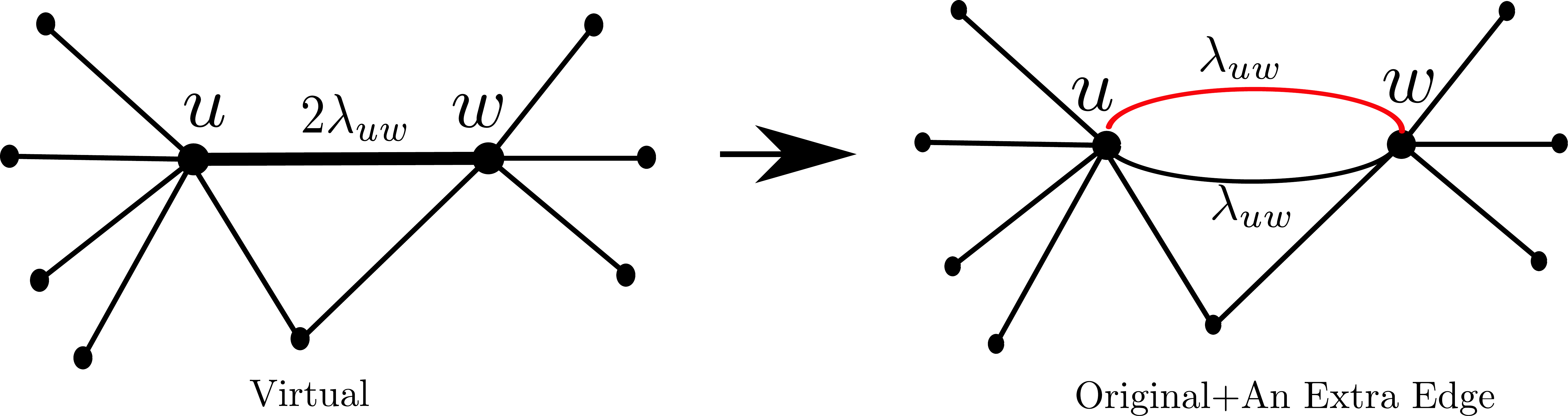} \hspace{0.4in}
\end{center}
\vspace{-0.3cm}
\caption{Interpretation of the virtual process using the original process and an extra edge.}
\label{fig:virtual-original}
\vspace{-0.2cm}
\end{figure}
Now, for these coupled processes, let us denote by $\tau^{\rm v}(x,y)$ a random variable denoting the first time when the random walkers in the virtual processes meet and $\tau^{\rm o}(x,y)$ the first time when the random walkers in the original process meet, given that the processes have been initialized from $x$ and $y$. We have that $\tau^{o}(x,y) = \tau^{\rm v}(x,y)$ if the first edge of rate $\lambda_{uw}$ registers an arrival before the second edge of rate $\lambda_{uw}$. This happens with probability $\frac{1}{2}$. If this does not happen, $\tau^{\rm o}(x,y)$ is $\tau^{\rm v}(x,y)$ plus another random variable which is a meeting time in the original process starting from $u$ and $w$, i.e., $\tau^{\rm o}(u,w)$. In other words, 
\begin{align}\nonumber
\tau^{\rm o}(x,y) = \begin{cases}
\tau^{\rm v}(x,y), &  \!\!\! \mbox{w.p}.\  \frac{1}{2} (\mbox{black edge ticks})\\
\tau^{\rm v}(x,y)+\tau^{\rm o}(u,w) &  \!\!\! \mbox{w.p}.\ \frac{1}{2} (\mbox{red edge ticks}).
\end{cases}
\end{align}
Taking expectation of the above relation, we obtain
\[ M^{\rm o}(x,y) \leq \frac{1}{2} M^{\rm v}(x,y) + \frac{1}{2} ( M^{\rm v}(x,y) + \max_{u,w} M^{\rm o}(u,w) ) \]
which implies that $\max_{x,y} M^{\rm o}(x,y) \leq 2 \max_{x,y} M^{\rm v}(x,y)$ which by Lemma \ref{thm:meeting-hitting-harmonic} is at most $12n^2$.
\end{proof}

\begin{remark} 
In view of the previous lemma, we have
\begin{equation} \label{expbound} \mathbb{P}(\tau^{\rm o}(x,y) > t) \leq e^{ - \lfloor \frac{t}{12en^2} \rfloor}.
\end{equation}This can be seen by the same argument as in Chapter 2.4.3 of \cite{aldous2002reversible} whose argument applies verbatim here.
\end{remark}

\begin{theorem}\label{thm:main-static-convergence-time}
The expected time until the quantized Metropolis dynamics reaches a consensus set is $O(n^2\log n)$.
\end{theorem}
\begin{proof} 
Consider the time it takes for the Lyapunov function $V(t) = \max_i x_i(t) - \min_i x_i(t)$ to shrink by at least 1 starting from any non-consensus configuration. We show that this time is $O(n^2 \log n)$, which proves the theorem.

Indeed, let $S_{\rm max}(t) = \arg \max_i x_i(t)$, $S_{\rm min}(t) = \arg \min_i x_i(t)$. Now let us consider the event that $V(t')=V(t)$ for $t' > t$. We claim that there must exist $i \in S_{\rm max}(t), j \in S_{\rm min}(t)$ such that two random walkers in the original process starting from nodes $i$ and $j$ have not yet met during the time interval $[t, t']$. To see this, we note that during the executions of the quantized Metropolis dynamics no new values $\max_i x_i(t)$ and $\min_i x_i(t)$ will be created; these values only travel from one node to the other over the network. Now given that at time $t'$ we have $V(t')=V(t)$, there must exist two nodes $i'$, $j'$ such that $x_{i'}(t')=\max_i x_i(t)$ and $x_{j'}(t')=\min_i x_i(t)$. By tracking back the origin of the random walks which bring the values $\max_i x_i(t)$ and $\min_i x_i(t)$ at time $t$ to the nodes $i'$ and $j'$ at time $t'$, one can see that these walks have never meet, and they must have originated from some nodes $i \in S_{\rm max}(t), j \in S_{\rm min}(t)$. Therefore, the probability that $V(t')=V(t)$ for $t' > t$ is upper bounded by the probability that there exist some $i \in S_{\rm max}(t), j \in S_{\rm min}(t)$ such that two random walkers in the original process starting from nodes $i$ and $j$ have not yet met. Thus letting $C$ being the time elapsed until $V(t)$ shrinks by 1, we have that 
\begin{align}\nonumber
\mathbb{E}(C)&\!=\!\int_{0}^{\infty}\!\!\!\mathbb{P}(C\!>\!t)dt\!\leq\! \int_{0}^{\infty}\!\!\!\!\min\left\{1,\sum_{x,y}\mathbb{P}(\tau^{\rm o}(x,y)>t)\!\!\right\}dt\cr
&\leq \int_{0}^{\infty}\min\left\{1,n^2e^{(1-\frac{t}{12en^2})}\right\}dt=O\left(n^2 \log n\right),
\end{align}where the second inequality is due to Eq. (\ref{expbound}) and the last equality follows from $\int_{0}^{\infty}\min\left\{1,Ae^{-at}\right\}dt=\frac{1+\log A}{a},$ which holds when $A \geq 1$, and is from \cite{aldous2002reversible}, Ch. 5.3.2. 
\end{proof}

\section{Convergence Time over Dynamic Networks}\label{sec:time-varying}

In this section we analyze the expected convergence time of the quantized Metropolis dynamics over time-varying networks and prove our main result, namely that this convergence time is essentially quadratic in the number of nodes. 

Toward this aim, let us consider a sequence of connected undirected networks $\mathcal{G}(t)=(V, \mathcal{E}(t)), t\ge 0$, over the same set of vertices $V$ which may change at discrete time instances, i.e., $\mathcal{G}(t)=\mathcal{G}(\floor{t}), \forall t\ge 0$. Thus our notation for the set of neighbors of $x$ will now be $N_x(t)$, which depends on $t$. In contrast to the previous section, we will now assume that each node always possesses a self-loop, i.e., $x \in N_x(t)$ for all $x \in V, t \geq 0$. We thus introduce the notation $N_x'(t)$ to denote the neighbors of node $x$ excluding itself: $N_x'(t) = N_x(t) \setminus \{x\}$. The degree of a node $x$, denoted by $d_x(t)$, will be referred to as the cardinality of $N_x'(t)$.

\subsection{Quantized Metropolis Model over Time-Varying Networks}\label{subsec:Quantized-Metropolis-Time-Varying}

Given a network $\mathcal{G}(t)$ at time $t\ge 0$, we associate a Poisson process with each edge $\{i,j\}$ of $\mathcal{G}(t)$ with a rate of $\lambda_{ij}(t) = 1/\max(d_i(t), d_j(t))$, i.e., the Metropolis weight corresponding to that edge at time $t$. When an edge $\{i,j\} \in \mathcal{E}(t)$ registers an arrival, we let the incident nodes update their values based on \eqref{eq:Quantize-rule}. Moreover, for each node $x$ and time $t$, the self-loop $(x,x)$ registers arrivals according to a Poisson process with rate of
\begin{align}\label{eq:self-loop-rate}
\lambda_{xx}(t)=1-\frac{1}{2}\sum_{x_i\in N_x'(t)}\lambda_{xx_i}(t).
\end{align}
Note that $\lambda_{xx}(t) \geq 0$ for all $x \in V, t \geq 0$, and the sum of the rates of the edges of $\mathcal{G}(t)$ is always equal to $|V|=n$. When a self-loop registers an arrival, no update is made.

\subsection{Preliminary Definitions and Relevant Results}
We next introduce the notions of the ``original'' and ``virtual'' processes over the time-varying graph sequence $\{\mathcal{G}(t)\}$.

\begin{definition} 
Consider two random walkers moving based on whether the activated edge in the quantized Metropolis location is incident to them or not. That is, if a random walker is at node $i$ at time $t$ when the edge $\{i,j\} \in \mathcal{G}(t)$ registers an arrival, the walker moves to node $j$. As before, we will refer to this as the original process, but note that the graph sequence is now time-varying. The meeting time $M_t^{\rm o}(x,y)$ for $x \neq y$ is the expected time until two random walkers, starting at locations $x$ and $y$ at time $t$, are both incident on an edge which registers an arrival. By convention we set $M_t^{\rm o}(x,x) = 0$ for all $x \in V, t \geq 0$.
\end{definition}

\begin{definition}
We define the virtual process to be identical to the original process except when the two walkers in the original process are each other's neighbors at nodes $u$ and $w$ for some $t\ge 0$, i.e., $u\in N_w(t)$. At this time in the virtual process we let the edges $\{r,s\}\in \mathcal{E}(t)$ register arrivals at rates $\mu_{rs}(t)$ defined as follows.  \begin{align}\label{eq:virtual-rates}
\mu_{rs}(t) = \begin{cases}
2\lambda_{uw}(t), & \mbox{if}  \ \{r,s\}=\{u,w\} \\
\lambda_{uu}(t)-\frac{\lambda_{uw}(t)}{2}, & \mbox{if}  \ \{r,s\}=\{u,u\}\\
\lambda_{ww}(t)-\frac{\lambda_{uw}(t)}{2}, & \mbox{if}  \ \{r,s\}=\{w,w\}\\
\lambda_{rs}(t), & \mbox{if}  \ \mbox{else}.\\
\end{cases}
\end{align} Note that $\mu_{rs}(t) \geq 0$ for all $r,s \in V, t \geq 0$. We refer to the meeting time of the virtual process starting from $x$ and $y$ at time $t$ as $\tau_t^{\rm v}(x,y)$ and its expectation $M_t^{\rm v}(x,y)$. By convention, we set $M^{\rm v}_t(x,x) = 0$ for all $x \in V, t \geq 0$. 

In words, the rate of the edge $\{u,w\}$ is doubled, while the rate of the self-loops is decreased to make sure that all rates still sum up to $n$.
\end{definition}

\subsection{Convergence Rate over Time-Varying Networks}
In this part, we state our main results for the quantized Metropolis dynamics over time-varying networks. Define $\Lambda_{xy}(t)$ to be
\begin{align}\label{eq:Lambda-x-y-expression}
\Lambda_{xy}(t)&\!=\!\!\!\!\!\sum_{x_i \in N_x'(t)}\!\!\!\!\!\lambda_{xx_i}(t)+\!\!\!\!\!\sum_{y_i \in N_y'(t)}\!\!\!\!\!\lambda_{yy_i}(t)\!=\!2\left(2\!-\!\lambda_{xx}(t)\!-\!\lambda_{yy}(t)\right).
\end{align}
Fix $t$ and let $T_1$ be the first time after $t$ that an edge registers an arrival in the virtual process. We then have

\begin{small}
\begin{align}\label{eq:expected-meeting-conditional}
&\!\!\!\!\!\!\mathbb{E} \left[\tau_t^{\rm v}(x,y)|T_1\right]=\left(T_1-t\right)+\left(1-\frac{\Lambda_{xy}(T_1)}{n}\right)M_{T_1}^{\rm v}(x,y)\cr
&\!+\!\!\!\!\!\!\!\!\sum_{x_i\in N_x'(T_1)}\!\!\!\!\!\!\frac{\lambda_{xx_i}(T_1)}{n}M_{T_1}^{\rm v}(x_i,y)+\!\!\!\!\!\!\sum_{y_i\in N_y'(T_1)}\!\!\!\!\!\!\frac{\lambda_{yy_i}(T_1)}{n}M_{T_1}^{\rm v}(x,y_i).
\end{align}\end{small}

Note that this equation holds regardless of whether $x$ and $y$ are neighbors, and this fact is the reason why we introduced the virtual
process in the first place. Next the following lemma shows that $T_1$ follows an exponential distribution with parameter $n$.

\begin{lemma}
Let us consider a process obtained by restarting a Poisson process of rate $n$ at every discrete time instant $k = 0,1,2,....$. The first time arrival $T_1$ of the new process thus generated will follow an exponential distribution with parameter $n$.
\end{lemma}
\begin{proof}
The proof follows by computing the cumulative function of $T_1$ and showing that $\mathbb{P}(T_1 \leq t) = 1-e^{-nt}, \forall t$.
\end{proof}

Continuing with \eqref{eq:expected-meeting-conditional} and since $T_1$ is exponential with parameter $n$, we further have that 
\begin{small}
\begin{align}\label{eq:integral-recursion}
M_t^{\rm v}(x,y)&=\int_{t}^{\infty}ne^{-n(t_1-t)}\mathbb{E}\left[\tau_{t}^{\rm v}(x,y)|T_1=t_1\right]dt_1\cr
&=\frac{1}{n}+\!\!\int_{t}^{\infty}\!\!\!\!ne^{-n(t_1-t)}\left(1-\frac{\Lambda_{xy}(t_1)}{n}\right)M_{t_1}^{\rm v}(x,y)dt_1\cr
&+\int_{t}^{\infty}\!\!\!\!\!\!\sum_{x_i\in N_x'(t_1)}\!\!\!\!\!\!ne^{-n(t_1-t)}\frac{\lambda_{xx_i}(t_1)}{n}M_{t_1}^{\rm v}(x_i,y)dt_1\cr
&+\int_{t}^{\infty}\!\!\!\!\!\!\sum_{y_i\in N_y'(t_1)}\!\!\!\!\!\!ne^{-n(t_1-t)}\frac{\lambda_{yy_i}(t_1)}{n}M_{t_1}^{\rm v}(x,y_i)dt_1.
\end{align}    
\end{small}

Now let us define $v(t)$ to be a column
vector of length $n(n-1)$ whose entries are the variables
$M_t^{\rm v}(x,y), \forall x\neq y$ in any order. It follows that we can write the above recursion for $M_t^{\rm v}(x,y)$ in
matrix form as
\begin{align}\label{eq:vector-form-v-without-change-variable}
v(t)&=\frac{1}{n}{\bf 1}\!+\!\int_{t}^{\infty}ne^{-n(t_1-t)}D\left(t_1\right)v\left(t_1\right)dt_1,
\end{align}
where $D(t) \in \mathbb{R}^{n(n-1) \times n(n-1)}$ is a matrix whose rows and columns we will index by $(x,y), x \neq y$ as 

\begin{align}\label{eq:D-matrix}
D_{(x,y)(r,s)}(t) = \begin{cases}
\frac{\lambda_{rx}(t)}{n}, & \!\!\mbox{if}  \ s=y, r\in N_x(t)\setminus\{ x,y \} \\
\frac{\lambda_{sy}(t)}{n}, & \!\!\mbox{if}  \ r=x, s\in N_y(t)\setminus\{ x,y \} \\
1-\frac{\Lambda_{xy}(t)}{n}, & \!\!\mbox{if}  \ (r,s)=(x,y) \\
0, &  \ \mbox{Else}.\\
\end{cases}
\end{align}

We change variables in \eqref{eq:vector-form-v-without-change-variable} from $t_1$ to $z=-e^{-n(t_1-t)}$ and obtain

\begin{align}\label{eq:vector-form-v}
v(t)&=\frac{1}{n}{\bf 1}\!+\!\int_{-1}^{0}\!\!D\left(t\!-\!\frac{\ln(-z)}{n}\right)v\left(t\!-\!\frac{\ln(-z)}{n}\right)dz.
\end{align}

We justify this change of variables by appealing to Theorem 263I in \cite{fremlinmeasure}. Indeed, the equivalence of Eq. (\ref{eq:vector-form-v}) and Eq. (\ref{eq:vector-form-v-without-change-variable}) is an instance of the equality $\int_{I'} g  = \int_{I} g(\phi(z)) \phi'(z)$. Here $\phi(z) = -e^{-n (z - t)}$ while $g(z) = D(t - \ln(-z)/n) v(t - \ln(-z)/n)$. Theorem 263I in \cite{fremlinmeasure} allows us to assert this equality subject to (i) $g(\cdot)$ being Lebesgue measurable (ii); and $\phi(z)$ being absolutely continuous on any closed bounded subinterval of $I$. Item (ii) is clearly satisfied here. Item (i) follows because $D(\cdot)$ is a piecewise continuous function by definition, and $v(\cdot)$ is a continuous function (indeed, taking $t_n \rightarrow t$ and conditioning on no transitions in $[t_n,t]$ as well as no meeting occurring in the same interval, we immediately obtain the continuity of $v(t)$ as a function of $t$. 

Next, we note that it is immediate from Eq.  \eqref{eq:D-matrix} that $D(t)$ is a symmetric matrix. Furthermore, it is easy to see that $D(t)$ is sub-stochastic,
implying that its eigenvalues are all real and less than $1$ in modulus. Our next step is to upper bound both the largest and smallest eigenvalues of
$D(t)$. 

The first step is to extend $D(t)$ to a stochastic matrix  $P \in \R^{n^2 \times n^2}$ as follows:
\begin{align}\label{eq:P-matrix}
P_{(x,y)(r,s)}(t)\!=\! \begin{cases} 1, \!\!& \mbox{if} \ x=y=r=s \\
\frac{\lambda_{rx}(t)}{n}, & \mbox{if}   \ x \neq y, s=y, r\in\! N_x(t)\!\setminus\!\{x\} \\
\frac{\lambda_{sy}(t)}{n}, & \mbox{if}  \ x \neq y, r=x, s\in\! N_y(t)\!\setminus\!\{y\} \\
1\!-\!\frac{\Lambda_{xy}(t)}{n}, & \mbox{if}  \ x \neq y, (r,s)=(x,y),  \\
0, & \mbox{if}  \ \mbox{else}.\\
\end{cases}
\end{align}

Note that since $P(t)$ is stochastic, it can be interpreted as the transition matrix of a Markov chain with absorbing states $S = \{(x, x): x \in V \}$. When we adopt the convention that states in $S$ correspond to the first $n$ rows of $P(t)$, we have that the
matrix $P(t)$ can be represented as

\begin{align}\label{eq:block-representation-C-D-P}
P(t)=\left( \begin{array}{cc}
I_{n\times n} & 0 \\
C(t) & D(t) \end{array} \right),
\end{align}
where $C(t)$ and $D(t)$ are, respectively, matrices of sizes $(n^2-n)\times n$ and $(n^2-n)\times(n^2-n)$.
 
\begin{lemma}\label{lemm:D-eigenvalue-upperbound} Consider a Markov chain with transition matrix $Q =\left( \begin{array}{cc}
I_{k \times k} & 0 \\
C & D \end{array} \right)$, which has the additional property that there is a path starting from any node and ending in the first
$k$ nodes. Furthermore, let us assume that $D$ is symmetric and let us denote its largest eigenvalue by
$\lambda_{\rm max}(D)$. Let $H_i$ denote the expected time until the chain is absorbed in $\{1, \ldots, k\}$ starting at node
$i$ and $H = \max_i H_i$. Then
\begin{align}\nonumber
\lambda_{\max}(D) \leq 1-\frac{1}{H}.
\end{align}
\end{lemma}
\begin{proof} 
Indeed, $1/(1-\lambda_{\rm max}(D))$ is a positive number which is an eigenvalue of $(I-D)^{-1}$ and is
consequently upper bounded by $||(I-D)^{-1}||_{\infty}$. But by \cite{seneta2006non}, Theorem 4.5, we have that
the sum of the entries in the $i$'th row of $(I-D)^{-1}$ is $H_i$.
\end{proof}

\begin{lemma}\label{lemm:alpha-max} 
For all $t \geq 0$, $ \lambda_{\rm max}(D(t)) \leq 1 - \frac{1}{6n^3}$.
\end{lemma}
\begin{proof}
Fix $t$ and let $H^{\rm P(t)}(\{x,y\})$ denote the hitting time to $S = \{ \{i,i\} ~|~ i \in V\}$ in $P(t)$ starting from $\{x,y\}$. It is immediate that when $x \neq y$,
\begin{align}\nonumber 
H^{\rm P(t)}(\{x,y\})&=\frac{n}{\Lambda_{xy}(t)} + \sum_{x_i \in N_x'(t)} \frac{\lambda_{xx_i}(t)}{\Lambda_{xy}(t)} H^{\rm P(t)}(\{x_i,y\}) \cr 
&\qquad+\sum_{y_i \in N_y'(t)} \frac{\lambda_{y y_i}(t)}{\Lambda_{xy}(t)} H^{\rm P(t)}(\{x, y_i\}).  
\end{align} 
Comparing this to Eq. (\ref{virtualmeet}) and noting that $H^{P(t)}(\{i,i\})=0$ for all $i \in V$, we obtain that $H^{P(t)}(\{x,y\})/n$ equals the expected meeting time in the virtual process on the graph $G(t)$ starting from $x$ and $y$. Appealing to Lemma \ref{thm:meeting-hitting-harmonic}, we can conclude that expected
hitting time in $P(t)$ to the first $n$ states is at most $6n^3$. Appealing to Lemma \ref{lemm:D-eigenvalue-upperbound} then completes the proof.
\end{proof}

Observe that the lower bound
\begin{equation} \label{lambda-lower} \lambda_{\rm min}(D(t)) \geq 1 - \frac{4}{n} \end{equation}
follows immediately by Gershgorin circles due to the observation that $\Lambda_{xy}(t) \leq 2$ for all $x,y \in V, t \geq 0$. We are now in a position to state and prove the main result of this section.

\begin{theorem}\label{thm:main-time-varying-convergence-time}
The expected time until the quantized Metropolis dynamics over time-varying connected networks
reaches the consensus set is  $O(n^2 \log^2(n))$.
\end{theorem}
\smallskip
\begin{proof} 
Note that we may iterate the recursion of Eq. \eqref{eq:vector-form-v} to obtain \begin{footnotesize}
\begin{align}\nonumber
v(t)&= \frac{1}{n} {\bf 1} + \int_{-1}^0 D \left( t - \frac{ \ln (-z)}{n} \right) \left( \frac{1}{n} {\bf 1} \right.\cr 
&\left.+  \int_{-1}^0 \!\!D \left( t \!-\! \frac{\ln(-z)}{n} \!-\! \frac{\ln (-z')}{n} \right) \!\cdot\! v \left( t \!-\! \frac{\ln(-z)}{n} \!-\! \frac{\ln (-z')}{n} \right)  dz' \right) \!dz  \cr 
&\leq \frac{2}{n} {\bf 1} + \int_{-1}^0 \int_{-1}^0 D \left( t - \frac{\ln (-z)}{n} \right) D \left( t - \frac{\ln(-z)}{n} - \frac{\ln (-z')}{n} \right)\cr 
& \qquad \cdot v \left( t - \frac{\ln(-z)}{n} - \frac{\ln (-z')}{n} \right)   ~ dz' dz,
\end{align}\end{footnotesize}  
$\!\!\!$where the last step used the sub-stochasticity of $D(\cdot)$. Iterating this $k$ times, we obtain
\begin{align}\label{eq:k-recursion-iteration}
v(t)\!\leq\! \frac{k}{n}{\bf 1}\!+\!\!\!\!\!\!\int\limits_{[-1,0]^k}\!\!\!\!\!\!\resizebox{0.02\hsize}{!}{$~$}\resizebox{0.6\hsize}{!}{$\left[\!\prod\limits_{j=1}^{k} \!D\!\left(\!t\!-\!\!\sum\limits_{i=1}^{j}\!\!\frac{\ln(-z_i)}{n}\!\!\right)\!\right]\!v\!\left(\!t\!-\!\!\!\sum\limits_{i=1}^{k}\!\!\frac{\ln(-z_i)}{n}\!\!\right)$}dz_k\resizebox{0.03\hsize}{!}{$\ldots$}dz_1.
\end{align} Now let us use introduce the notation $M=\sup_{t\ge 0}\|v(t)\|_{\infty}$. Elementary arguments suffice to establish that $M<\infty$; recall that $v(t)$ stacks up expected meeting times in the virtual process starting at time $t$, and the finiteness of $M$ can follow from the observation that the probability of intersection in any positive length interval $[t,t+a]$ is positive and bounded away from zero independently of $t$. Thus for every $\epsilon>0$ there exists $t^*$ such that $M-\epsilon\leq \|v(t^*)\|_{\infty}$. Since the entries of $D(\cdot)$ and $v(\cdot)$ are always nonnegative, we can take the infinity-norm from both sides of \eqref{eq:k-recursion-iteration} and use triangle inequality to obtain
\begin{align}\label{eq:infinit-norm-to-2-norm}
\!&M-\epsilon \leq \|v(t^*)\|_{\infty}\leq \frac{k}{n}\cr
&\!\!+\!\!\int\!\!\!\!\int\!\!\resizebox{0.03\hsize}{!}{$\ldots$}\!\!\!\int\resizebox{0.75\hsize}{!}{$\left\|\prod\limits_{j=1}^{k} \!D\!\left(\!t^*\!-\!\!\sum\limits_{i=1}^{j}\!\!\frac{\ln(-z_i)}{n}\!\!\right)\right\|_{\infty}\!\left\|v\!\left(\!t^*\!-\!\!\!\sum\limits_{i=1}^{k}\!\!\frac{\ln(-z_i)}{n}\!\!\right)\right\|_{\infty}$}\!\!\!\!dz_k\resizebox{0.03\hsize}{!}{$\ldots$}dz_1\cr
&\!\leq \frac{k}{n}+M\!\! \int\!\!\!\!\int\!\!\resizebox{0.03\hsize}{!}{$\ldots$}\!\!\!\int\resizebox{0.45\hsize}{!}{$\left\|\prod\limits_{j=1}^{k} \!D\!\left(\!t^*\!-\!\!\sum\limits_{i=1}^{j}\!\!\frac{\ln(-z_i)}{n}\!\!\right)\right\|_{\infty}$}\!\!\!dz_k\resizebox{0.03\hsize}{!}{$\ldots$}dz_1\cr
&\!\leq\frac{k}{n}+M\!\! \int\!\!\!\!\int\!\!\resizebox{0.03\hsize}{!}{$\ldots$}\!\!\!\int\resizebox{0.6\hsize}{!}{$\!\!\sqrt{n(n-1)}\left\|\prod\limits_{j=1}^{k} \!D\!\left(\!t^*\!-\!\!\sum\limits_{i=1}^{j}\!\!\frac{\ln(-z_i)}{n}\!\!\right)\right\|_{2}$}\!dz_k\resizebox{0.03\hsize}{!}{$\ldots$}dz_1\cr
&\!\leq\frac{k}{n}+M\!\! \int\!\!\!\!\int\!\!\resizebox{0.03\hsize}{!}{$\ldots$}\!\!\!\int\resizebox{0.6\hsize}{!}{$\!\!\sqrt{n(n-1)}\prod\limits_{j=1}^{k} \!\left\|D\!\left(\!t^*\!-\!\!\sum\limits_{i=1}^{j}\!\!\frac{\ln(-z_i)}{n}\!\!\right)\right\|_{2}$}\!dz_k\resizebox{0.03\hsize}{!}{$\ldots$}dz_1\cr
&\!\leq \frac{k}{n}+M\!\sqrt{n(n-1)}\left(\sup_{t\ge 0}\left\{|\lambda_{\max}(D(t))|, |\lambda_{\rm min}(D(t))|\right\}\right)^{k}\cr
&\!\le \frac{k}{n}+M\!\sqrt{n(n-1)}\left(1-\frac{1}{6n^3}\right)^{k},
\end{align}
where in the fourth inequality we have used the fact that the infinity-norm of a matrix is always upper bounded by its induced 2-norm times the square root of its dimension, i.e., $\|D(\cdot)\|_{\infty}\leq \sqrt{n(n-1)}\|D(\cdot)\|_2$. Moreover, the last inequality is valid due to Lemma \ref{lemm:alpha-max} and Eq. (\ref{lambda-lower}).

Let us choose $k^*$ so that $\sqrt{n(n-1)} \left( 1 - \frac{1}{6n^3} \right)^{k^*} \leq \frac{1}{2}.$ Appealing to the inequality $(1-1/x)^{x} \leq e^{-1}$, we obtain that we may choose $k^* = O(n^3 \log n)$ to accomplish this. Plugging this into the last line of Eq. (\ref{eq:infinit-norm-to-2-norm}), we immediately obtain $M - \epsilon \leq \frac{k^*}{n} + \frac{M}{2}$. Since this holds for any $\epsilon > 0$, it implies that $M \leq 2 k^*/n = O(n^2 \log n)$.

We have thus obtained an upper bound of $O(n^2 \log n)$ on the expected meeting time in the virtual process starting from any pair of nonidentical nodes and any time $t$. The rest of the proof directly parallels the arguments we have made in the fixed graph case, and we only sketch it rather than repeat the relevant parts verbatim. The first step is to argue that $2M$ is an upper bound on the meeting time of the original process. The proof proceeds exactly as in the case of Lemma \ref{orig-lemma} by coupling the two processes and conditioning on the transition at the last step in the meeting time of the virtual process.  The next (and final) step is to argue that the time it takes $V(t) = \max_i x_i(t) - \min_i x_i(t)$ to shrink by $1$ is upper bounded by $O(n^2 \log^2 n)$. This follows exactly as in the proof of Theorem \ref{thm:main-static-convergence-time}. Indeed, we can argue that if $V(t')=V(t)$ at some time $t'>t$, this means that a pair of nodes performing a random walk according to the original process have not met between times $t$ and $t'$. Upper bounding the latter using the union bound, we once again obtain that the expected time it takes for $V(t)$ to shrink by 1 is the expected meeting time in the virtual process times a multiplicative factor of $O( \log n)$. 
\end{proof}


\section{Conclusion}\label{sec:conclusion}
We have studied the quantized consensus problem on undirected connected networks in both static and time-varying settings. In particular, we have proved an upper bound of $O(n^2 \log^2 n)$ for the convergence time of quantized Metropolis over time-varying connected networks.

A future direction of research would be to improve convergence times further. For example, \cite{linear} attained a linear convergence time for consensus on any fixed graph, and it is an open question to obtain a quantized consensus protocol which replicates this. Moreover, an interesting problem is to see to what extent the results here can be carried over to protocols with nonlinear transmission \cite{chen2013consensusnonlinearity}.

\bibliographystyle{IEEEtran}
\bibliography{thesisrefs}
\end{document}